\documentclass[11pt,draft]{amsart}
\usepackage{amsthm,amssymb,amsmath,amscd,amsfonts,enumitem,graphicx,tikz,hyperref}
\hypersetup{final}
\usepackage[utf8]{inputenc}
\usepackage[english]{babel}
\usepackage[margin=1in]{geometry}

\theoremstyle{plain}
\newtheorem{thm}{Theorem}[section]
\newtheorem{lem}[thm]{Lemma}
\newtheorem{prop}[thm]{Proposition}

\theoremstyle{definition}
\newtheorem{defn}{Definition}[section]

\DeclareMathOperator*{\ulalim}{\underleftarrow\lim}

\DeclareMathOperator{\rank}{rank}

\makeatletter
\let\@@pmod\pmod
\DeclareRobustCommand{\pmod}{\@ifstar\@pmods\@@pmod}
\def\@pmods#1{\mkern4mu({\operator@font mod}\mkern 6mu#1)}
\makeatother

\newcommand{\mbbC}{\mathbb{C}}

\newcommand{\mbbF}{\mathbb{F}}

\newcommand{\mbbR}{\mathbb{R}}

\newcommand{\mbbZ}{\mathbb{Z}}

\author{Kevin T. Tian}
\author{Eric Samperton}
\address{Department of Mathematics, University of California, Santa Barbara, CA 93106, USA}
\email{ktian@math.ucsb.edu;eric@math.ucsb.edu}

\author{Zhenghan Wang}
\address{Microsoft Station Q and Department of Mathematics, University of California, Santa Barbara, CA 93106, USA}
\email{zhenghwa@microsoft.com;zhenghwa@math.ucsb.edu}

\begin{document}

\title{Haah codes on general three-manifolds}

\thanks{K.T. and Z.W. are partially supported by NSF grant  FRG-1664351. Z.W. would like to thank J. Haah and M. Hastings for insightful discussions.}
\begin{abstract}

Haah codes represent a singularly interesting gapped Hamiltonian schema that has resisted a natural generalization, although recent work shows that the closely related type I fracton models are more commonplace.  These type I siblings of Haah codes are better understood, and a generalized topological quantum field theory framework has been proposed.  Following the same conceptual framework, we outline a program to generalize Haah codes to all 3-manifolds using Hastings' LR stabilizer codes for finite groups.

\end{abstract}


\date{\today}
\maketitle

\section{Introduction}

Haah codes represent a singularly interesting gapped Hamiltonian schema that has resisted a natural generalization, although recent work shows that the closely related type I fracton models are more commonplace \cite{Ha11,vhf15,NH18}.  These type I siblings of Haah codes are better understood, and a generalized topological quantum field theory (TQFT) framework is proposed in \cite{SSWC18}.  Following the same conceptual framework in \cite{SSWC18}, we outline a program to generalize Haah codes to all 3-manifolds using Hastings' LR stabilizer codes for finite groups \cite{Hs14}.

Fracton models challenge the conventional notion of topological phase wherein the finite dimensional ground state Hilbert space $V(Y)$ depends only on the topology of the spatial manifold $Y$ that the phase occupies.  In fact, to be more precise, some subtlety already reared its head in conventional topological phases of fermion systems and Chern-Simons theories.  For fermion systems, the associated Hilbert space of an oriented spatial manifold $Y$ also depends on a spin structure $s$, so the Hilbert space is only well-defined for the pair $(Y,s )$.  Likewise, the framing anomaly for Chern-Simons theories means that the Hilbert space is only well-defined for a pair $(Y,f)$, where $f$ is a framing of $Y$ refining a spin structure.  Inspired by these ideas, \cite{SSWC18} argues that to formulate the $X$-cube model as a generalized TQFT, $Y$ should be equipped with a \emph{singular compact total foliation,} which is a kind of dual framing---three sets of perpendicular integrable tangent planes at each point.

In this paper, we propose that in order to generalize the Haah code to all 3-manifolds,
the correct additional topological structure on the spatial manifold $Y$ is a pair of finite subsets $(S_1,S_2)$ of the fundamental group $\pi_1(Y)$ of the spatial manifold $Y$ (we always assume $Y$ is connected and omit the base point as it is immaterial for our discussion).  The two subsets $\{S_i\}$ define a \lq\lq unit cell" $U_e$, which is the cone from the unit element to all the elements in $\{S_1\cup \bar{S_1}\cup \bar{S_2}\cup S_2\}$ for the fundamental group lattice in the universal cover $\widetilde{Y}$ of $Y$.  Choosing $S_1=\{1,x,y,z\},S_2=\{1,xy,yz,xz\}$ of $\mbbZ^3$ recovers the Haah codes for $T^3$.  The resulting theory typically assigns an infinite dimensional Hilbert space $V(Y,S_1,S_2)$ to each spatial manifold $Y$ equipped with a pair $(S_1,S_2)$.  The ground state Hilbert space $V(Y,S_1,S_2)$ is constructed from the profinite completion $\widehat{\pi_1(Y)}$ of the fundamental group $\pi_1(Y)$.  Fractal pictures arise when the profinite completion $\widehat{\pi_1(Y)}$ is visualized.

The fundamental group $\pi_1(X)$ of a topological space $X$ is one of the most important topological invariants.  It lives two dual lives: it can be identified as a collection of points $\Pi(X)$ in the universal cover $\widetilde{X}$ or as a collection of based closed loops in $X$.  A well-known example is the $n$-torus $T^n$, which has fundamental group $\mathbb{Z}^n$ and universal cover $\mathbb{R}^n$.  When $X$ is a manifold, the collection of points $\Pi(X)$ is homogeneous in the sense that their small neighborhoods are homeomorphic copies of the same open ball. Therefore, we can regard $\Pi(X)$ as a \lq\lq lattice" with one site for each point and edges will be generated by translations of the unit cell $U_c$ to $U_g$, which is the cone from $g$ to $\{gS_1\cup g\bar{S_1}\cup \bar{S_2}g\cup S_2g\}$\footnote{There are other natural ways to construct the edges, but they are not essential for our theory.}. Our models are then spin models on such infinite lattices $\pi_1(Y)$ approximated by its finite quotients.  In our generalization, the fundamental group $\pi_1(Y)$---identified as a lattice $\Pi(Y)$ in $\widetilde{Y}$---is subjected to periodic boundary conditions, which are labeled by finite index normal subgroups (FINs) of $\pi_1(Y)$.  The corresponding finite lattices ``in $Y$" are the closed covering $3$-manifolds $\widetilde{Y}_N$ with fundamental groups $N$ corresponding to the boundary conditions $N$.

Our generalized Haah models can be considered either as inside the space manifold as a collection of loops or as outside the space manifold as a collection of points in the universal cover.  We define a profinite low energy limit by a limiting procedure outside the manifold, so the resulting limit is closer to a large volume limit.  We did not investigate limiting procedures for loops inside the space manifold, which would be closer to a scaling limit.  Manifolds such as the $3$-torus $T^3$ and $S^1\times S^2$ are homeomorphic to all its covers, hence for these two manifolds the inside and outside point of views can be confusing.

Though the abstract theory applies to all groups, the existence of an appropriate limit would put strong restrictions on the groups $\Gamma$.  Our models should have limits on the profinite completion $\widehat{\Gamma}$ of $\Gamma$, a condition which makes fundamental groups of 3-manifolds especially pertinent.  Indeed, 3-manifolds stand out because their fundamental groups are always residually finite\footnote{A group is \emph{residually finite} if the intersection of all of its FINs is the trivial subgroup.} \cite{T82,He87}, and a group $\Gamma$ injects into its profinite completion $\widehat{\Gamma}$ if and only if $\Gamma$ is residually finite.

It should also be interesting to consider other groups associated to a topological space $X$, such as the higher homotopy groups $\pi_i(X)$ or homology groups $H_i(X)$.  Higher homotopy groups are preserved under covering spaces, therefore, they might make our scaling scheme easier.  Homology, on the other hand, does not behave trivially under covering spaces in any dimension, although in this case there are nice module structures coming from the deck transformation groups.  When the fundamental groups are abelian (such as in the Haah codes), then they are the same as the first homology $H_1$.  For concreteness and to maintain a close connection to Haah codes, our focus will be on the fundamental groups of closed (that is, compact without boundary) $3$-manifolds in this paper.

Another motivation of our work is to bring geometric group theory into condensed matter physics.  Our program will only be outlined in this paper, and the details will appear in \cite{STW}.

\section{Manifolds, Lattices, and Periodic Boundary Conditions}

Haah codes are defined on the cubic lattice $\mbbZ^3$ with periodic boundary conditions \cite{Ha11}.  Our interpretation is that the Haah codes are defined on lattices inside the $3$-torus $T^3$, whose fundamental group is $\mbbZ^3$.  To generalize Haah codes to all closed $3$-manifolds $Y$ we should understand: what do lattices and periodic boundary conditions mean in general?

Manifolds arose historically as domains of functions more general than regions in the Euclidean space $\mathbb{R}^n$.  Periodic functions are simply functions defined on the circle $S^1$, while doubly periodic functions are functions on the $2$-torus $T^2$.  More generally, functions on a closed manifold $Y$ will be regarded as some generalized periodic functions.  Each closed manifold can be realized as a convex polyhedron in Euclidean space with some complicated identification of the boundary with itself\footnote{Our manifolds in this paper are smooth, so by Whitehead's theorem they admit combinatorial triangulations.  Removing a top simplex from the triangulation of a given manifold, we collapse the rest of the manifold onto a spine, which is the gluing pattern of the polyhedron consisting of the union of the simplices in the collapsing sequence.}, hence our point-view agrees with the usual notion of a periodic function.  For example, every genus $g$ surface can be formed by gluing the edges of a convex $4g$-regular polygon.

The cubic lattice $\mbbZ^3$ in Haah codes is naturally replaced by the fundamental group $\pi_1(Y)$ of the manifold $Y$, visualized as a collection of $\pi_1(Y)$-invariant points in the universal cover $\widetilde{Y}$  of $Y$.  We propose that a periodic condition for $\pi_1(Y)$ is simply a finite index normal subgroup (FIN) $N \subset \pi_1(Y)$.  For the Haah codes, a normal subgroup of $\mbbZ^3$ is of the form $L_1\mbbZ \times L_2\mbbZ \times L_3\mbbZ$, which is exactly the usual periodic boundary conditions.  Note that the Haah codes are defined on the quotient groups $\mbbZ/{L_1\mbbZ}\times \mbbZ/{L_2\mbbZ}\times \mbbZ/{L_3\mbbZ}=\mbbZ^3/(L_1\mbbZ \times L_2\mbbZ \times L_3\mbbZ)$.

There are two kinds of lattices in Haah codes: the infinite cubic lattice $\mbbZ^3$ in the universal cover $\mbbR^3$, and the finite lattices $\mbbZ/L_1\mbbZ \times \mbbZ/L_2\mbbZ \times \mbbZ/L_3\mbbZ$ in the $3$-torus $T^3$.  In our generalization, the infinite cubic lattice $\mbbZ^3$ is generalized to the fundamental group $\pi_1(Y)$, and the finite lattices labeled by FINs $N$ are the quotient finite groups $G_N=\pi_1(Y)/N$.  The finite lattice $G_N$ labeled by a FIN $N$ is visualized as the finite covering space ${\widetilde{Y}}_N$ of $Y$ whose fundamental group is $N$ and where the relevant group of translations is the deck transformation group $G_N$ for the cover $\widetilde{Y}_N \to Y$.  In fact, all of $\pi_1(Y)$ acts by translations on $\widetilde{Y}_N$, but the kernel is precisely $N$.  This dual correspondence between FINs of $\pi_1(Y)$ and finite regular covering spaces of $Y$ is understood more generally via the classification of covering spaces; see \cite[Ch.~1]{Ht}.

We will use the following notations and terminologies:
\begin{itemize}
    \item $G$ denotes a finite group.
    \item $\Gamma$ denotes a general group, not necessarily finite or even countable.  In general, we denote the unit element $e\in \Gamma$ as $1$, however, when $\Gamma=\mathbb{Z}$ or some other abelian group in additive notation, we write $e=0$.
    \item $\{S_i,i=1,2,...\}$ some finite subsets of $\Gamma$ or $G$.
    \item $\mathcal{N}(\Gamma)$ denotes the set of all FINs of $\Gamma$.  Observe that $\mathcal{N}(\Gamma)$ has the structure of a partially ordered set (poset) if we order FINs by inclusion.
    \item $N,M$ denote particular FINs.
\end{itemize}

In condensed matter physics, a many-body quantum system on a space $Y$ has a preferred coordinate system given by the ``lattice" of ``atoms" (or ``spins").  The word ``atom" is used as a catch-all for any local constituent of the many-body system consisting of a cluster of things that, taken together, are regarded as a local degree of freedom (LDOF).  We will consider only uniform LDOFs consisting of a few qudits $(\mbbC^d)^{\otimes q}$.  The word ``lattice" is also generalized to be used interchangeably with ``graph" so that atoms are located at vertices of the graph.  Often graphs in this paper will be some Cayley graphs of finite groups, so a generalized version of translation invariant lattice makes sense.  Identifying group elements of $\Gamma$ as the sites of a graph, we imagine there is an associated Hilbert space $L(\Gamma, d, q)$ for all degree of freedom (DOF); formally the Hilbert space is
\[ L(\Gamma, d,q)=\otimes_{\gamma \in \Gamma} ((\mbbC^d)^{\otimes q})_\gamma. \]
A mathematically rigorous definition of $L(\Gamma, d,q)$ is subtle when $\Gamma$ is infinite, so we defer the discussion until the last section.

\section{LR Models on finite groups}

Haah developed the polynomial method to study Pauli Hamiltonians or stabilizer codes on translation-invariant lattices from abelian groups \cite{Ha16,Ha13}.  This powerful method translates the study of degeneracy and excitations of such models into the mathematics of symplectic geometry over finite fields, and some of the ideas extend to arbitrary (non-abelian) groups.  In this section, we define the LR model on finite groups based on Hastings' LR codes \cite{Hs14}, and then apply Haah's theory to study the degeneracy of the LR Hamiltonians on general finite groups.

\subsection{LR Hamiltonian schema}

A Hamiltonian schema means a recipe to construct families of Hamiltonians from some given input data.

\subsubsection{Hastings' LR codes}

The third author learned about LR codes on finite groups through a private communication with Hastings that contained the material in this subsection \cite{Hs14}.  We present the LR codes as a Hamiltonian schema on finite groups.

Let $(G;S_1,S_2)$ be a triple where $G$ is any finite group, and $S_1$ and $S_2$ are two fixed subsets of $G$.   The two subsets $S_1, S_2$ are used to translate any fixed element of the group $G$.  Let
\[ \begin{aligned}
\bar{S}&=\{h^{-1}|h\in S\}, \\
gS&=\{gh|h\in S\}, \\
Sg&=\{hg|h\in S\},
\end{aligned} \]
for any $g\in G, S\subset G$.

Let
\[ L(G,2,2)=\bigotimes_{g\in G} (\mbbC^2\otimes \mbbC^2)_g \]
be the Hilbert space that assigns a bi-qubit $\mbbC^2\otimes \mbbC^2$ to each group element $g$. We may imagine there are two copies of $G$ in two layers $G\times \{\pm\}$, and the first and second qubits of the bi-qubit $(\mbbC^2\otimes \mbbC^2)_g$ are assigned to $(g,+)$ and $(g,-)$, respectively.  Given a set $S\subset G$ and a Pauli matrix $P$ ($P=X,Y,Z$ for Pauli $\sigma_i,i=x,y,z$), we denote by $P_{(S, \epsilon)}$ the product matrix of $P$ acting on each qubit labeled by $(g,\epsilon), g\in S, \epsilon=\pm $.  Let
\[ Z_g=Z_{(gS_1,+)}\cdot  Z_{(\bar{S_2}g,-)}, \]
\[ X_g=X_{(S_2g,+)}\cdot X_{(g\bar{S_1},-)}\]
be two stabilizers on each bi-qubit. More explicitly,
\[\begin{aligned}
  Z_g =& \prod_{v\in S_1, w\in S_2} Z_{(gv,+)}\cdot Z_{(w^{-1}g, -)} \\
  X_g =& \prod_{v\in S_1, w\in S_2} X_{(wg, +)}\cdot X_{(gv^{-1}, -)}
\end{aligned}. \]

Then the LR Hamiltonian on $L(G,2,2)$ is defined as
\[ H(G;S_1,S_2)=\sum_{g\in G}\frac{I-Z_g}{2}+\sum_{g\in G}\frac{I-X_g}{2}.\]

\begin{prop} \label{prop:hamiltonian}

\begin{enumerate}

\item Any two stabilizers in $\{Z_g, X_h \mid g,h\in G \}$ commute with each other.
\item $\prod_g Z_g=\prod_g X_g=I$ if both sets $S_i,i=1,2$ are of even parity, i.e. have an even number of elements.
\item If both sets $S_i,i=1,2$ have even parity, then the ground state subspace of $H(G;S_1,S_2)$ is degenerate, i.e. the lowest energy eigenspace has dimension more than one\footnote{The dimension of the ground state manifold is then referred to as the degeneracy.}.

\end{enumerate}

\end{prop}

To prove (1), we simply observe that when some pair $Z_g, X_h$ does not commute, then they anti-commute and $gv=wh$ for some pair $v\in S_1, w\in S_2$.  Then there would be a corresponding anti-commuting pairs $w^{-1}g=hv^{-1}$. Hence  non-commuting pairs appear an even number of times.

The proof of (2) is a moment's thought, and (3) follows from (2).

There is also a generalization to pairs of qudits.  Let $U^a,V^a, a=1,2$, be two operators acting on the $a$-qudit $\mbbC^{d}$ with eigenvalues $\{\omega_d^m\}, m=0,1,...,d-1, \omega_d=e^{\frac{2\pi i}{d}}$, and $U^aV^a=\omega_d V^aU^a$.  Suppose $m_a$ are two functions on $S_a$ with values in $\{1,2,..,d-1\}$, respectively.  Then the LR model on $L(G,2,2)$ can be generalized to qudits $L(G,d,2)$ by replacing $Z_g$ and $X_g$ with
$$A_g=\prod_{v\in S_1, w\in S_2}(U^1_{gv})^{m_1(v)}(U^2_{w^{-1}g})^{m_2(w)},$$ and $$B_g=\prod_{v\in S_1, w\in S_2}(V^1_{wg})^{m_2(w)}(V^2_{gv^{-1}})^{-m_1(v)},$$
respectively.

\subsection{Degeneracy}
Given a collection of qubits $L(V,2,1)=\otimes_{v\in V}\mbbC^2$ indexed by a finite set $V$, a Pauli operator on $L(V,2,1)$ is an operator that is a tensor product of Pauli matrices $\{I,X,Y,Z\}$ acting on the single qubits (note that we include the identity $I$ as a Pauli matrix here).  Up to phases $\pm 1, \pm i$, the Pauli operators can be parameterized by vectors in $\mbbF_2[V_X]\oplus \mbbF_2[V_Z]$, where $V_X,V_Z$ are two copies of $V$ as follows and $\mbbF_2=\mbbZ/2\mbbZ=\{0,1\}$ is the field with two elements.  An $\mbbF_2$-vector $v=(v_i) \in \mbbF_2[V_X],v_i\in \mbbF_2$ represents the Pauli operator $P_v$ acting on the $i^\text{th}$ qubit by $X^{v_i}$, i.e.
\[ P_v = \bigotimes_{i \in V} X^{v_i}. \]
Similarly, an $\mbbF_2$-vector in $\mbbF_2[V_Z]$ represents a Pauli operator acting on qubit $i$ by $Z^{v_i}$.  Then any Pauli operator on $L(V,2,1)$, up to phases, can be written as a composition of $P_v$ and $P_w$ for $\mbbF_2$-vectors $v \in \mbbF_2[V_X], w \in \mbbF_2[V_Z]$, hence can be represented by the vector $(v,w)\in \mbbF_2[V_X]\oplus \mbbF_2[V_Z]$.

Let $V=G \times \{+,-\}$.  Then we will represent the stabilizers $X_g, Z_g$ as vectors in $\mathbb{F}_2[G \times \{+,-\}]^2$.

Any stabilizer code defined by $r$ independent stabilizers on $n$ qubits has code space dimension $2^{n-r}$.  The number of independent stabilizers is the rank of the following matrix whose columns are these vectors.  This matrix can be written in the following form:

\[M_G := \left(
    \begin{array}{c|c}
    S_2g & \\[0.4em]
    g\overline{S_1} \\ \hline
    & gS_1 \\[0.4em]
    & \overline{S}_2g\\
    \end{array}
  \right)\]
where each block is $|G| \times |G|$.

The first $|G|$ columns represent the $X_g$ operators as $g$ varies through $G$ and the second $|G|$ columns likewise represent the $Z_g$ operators.  We have the following obvious statement:

\begin{prop}\label{prop:degeneracy}

If the $(4|G|)\times (2|G|)$ matrix $M_G$ over $\mathbb{F}_2$ has rank $k$, then the degeneracy of the LR Hamiltonian is $2^{2|G|-k}$.
The rank $k$ is always an even integer.

\end{prop}

\subsection{Examples of degeneracy}

\subsubsection{Quartic interactions}
We begin with a result for a rather general family of LR models---those with $|S_1|=|S_2|=2$.

\begin{prop}
Consider a LR model where $G$ is any finite group and $S_1=\{1,s\}, S_2=\{1,t\}$, where $s,t$ are nontrivial elements of $G$.  The ground state degeneracy of $H_G$ is
\[ 4^{|\langle t \rangle \backslash G/\langle s \rangle|}, \]
where $\langle t \rangle \backslash G/\langle s \rangle$ is the set of double cosets.  In particular, there are $2\cdot|\langle t \rangle \backslash G/\langle s \rangle|$ logical qubits, and
\[ |\langle t \rangle \backslash G/\langle s \rangle| = |G/\langle s,t\rangle| \]
if $\langle s \rangle$ is a normal subgroup of $G$.
\label{prop:1s1t}
\end{prop}

\begin{proof}
In order to apply Proposition \ref{prop:degeneracy}, we need to compute the rank of the matrix $M_G$.  We first show
\[ \rank \left(\begin{array}{c} gS_1 \\[0.5em] \overline{S_2}g \end{array}\right) = |G| - |\langle t \rangle \backslash G/\langle s \rangle| \]
It will be clear from a similar argument that
\[ \rank \left(\begin{array}{c} S_2g \\[0.5em] g\overline{S_1} \end{array}\right) = |G| - |\langle s \rangle \backslash G/\langle t \rangle| = |G|-|\langle t \rangle \backslash G/\langle s \rangle|,\]
whence the proposition follows.

Let $g \in G$ and consider the column corresponding to $Z_g$.  We abuse notation and continue to denote this column by $Z_g$.  We show that there is a linear dependence supported on the columns indexed by elements in the double coset $\langle t\rangle \backslash g / \langle s \rangle$, that is
\[\sum_{h \in \langle t\rangle \backslash g / \langle s \rangle} Z_h = 0. \]

Most of the entries of the $Z_g$ column are 0, but exactly 4 are nonzero:
\[ (g1,+1),(gs,+1),(1g,-1), \text{ and } (t^{-1}g,-1). \]
By our choice of $S_1$ and $S_2$, every row has precisely two entries that are a 1, so if $Z_g$ is involved in a dependence, there is precisely one other column $Z_h$ that can cancel a given nonzero entry of $Z_g$.  So, for example, since $Z_g$ has a nonzero entry at $(gs,+1)$, it must be cancelled by some other column.  The only possibility is $Z_{gs}$.  Similarly, since $Z_{gs}$ has a nonzero entry at $(gs^2,+1)$, it must be canceled by $Z_{gs^2}$.  Continuing in this way, we see that $Z_{gs^k}$ must be in the dependence for all $k=1,\dots,|s|$.  Moreover, it's clear that the nonzero entries of these $Z_{gs^k}$ in the +1 layer cancel in pairs.

On the other hand, since $Z_{gs^k}$ also has a nonzero entry at $(t^{-1}gs^k,-1)$, this must be canceled by $Z_{t^{-1}gs^k}$, whose nonzero entry at $(t^{-2}gs^k,-1)$ must be canceled by $Z_{t^{-2}gs^k}$, and we continue as before to get cancellation in pairs among the nonzero entries in the $-1$ layer of the $Z_{t^{-j}gs^k}$, where $k$ is fixed and $j=1,\dots,|t|$.  Returning to the +1 layer, the same argument as above shows that for a fixed $j$, the nonzero entries of the $Z_{t^{-j}gs^k}$ as $k$ varies through $1,\dots,|s|$ cancel in pairs, so we do not need to introduce any more columns to guarantee everything cancels.   This shows that every column is involved in a linear dependence supported on the columns in the double coset $\langle s\rangle \backslash g / \langle t \rangle$.

In fact, our argument also shows that any linear dependence that involves the column $Z_g$ must involve all of the columns in the double coset represented by $g$.  Thus, there are $|\langle s\rangle \backslash G / \langle t \rangle|$ minimal linear dependencies, where \emph{minimal} means that the linear dependence has support as small as possible.  The double cosets partition $G$, so the minimal dependencies do not share any support, and, hence, are themselves independent.  This shows
\[ \rank \left(\begin{array}{c} gS_1 \\[0.5em] \overline{S_2}g \end{array}\right) = |G| - |\langle t \rangle \backslash G/\langle s \rangle| \]
as claimed.
\end{proof}

We unpack this proposition in some specific examples below.

\subsubsection{Cyclic Groups}
Let us consider LR codes on finite cyclic groups $G=\mbbZ/n\mbbZ$, which we can interpret as translation-invariant finite discretizations of the circle $S^1$.

First, fix $a,b \in G$ nontrivial group elements, and consider an LR code with $S_1=\{0,a\}, S_2=\{0,b\}$.  Then, by Proposition \ref{prop:1s1t}, the ground state degeneracy of $H_G$ is  $4^{\gcd(a,b,n)}$.  See Figure \ref{f:example} for an example.

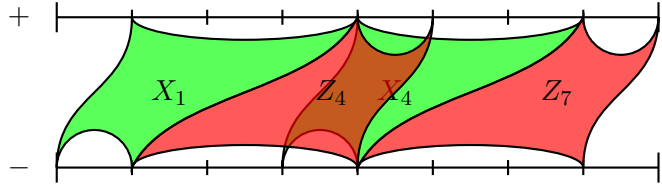
\begin{figure}
\begin{tikzpicture}
\draw[thick] (0,+1) -- (8,+1);
\draw[thick] (0,-1) -- (8,-1);
\foreach \i in {1,2,...,7}
	{\draw[thick] (\i,+1.1) --(\i,+0.9);
	\draw[thick] (\i,-1.1) --(\i,-0.9);
	}
	
\foreach \i in {1,4}
	{
	\draw[thick,fill=green,fill opacity=0.65] (\i,+1) arc[x radius = 1.5, y radius=0.3, start angle = 180, end angle = 360] .. controls (\i+2.5,0) and (\i+0.5,0).. (\i,-1) arc[x radius = 0.5, y radius = 0.5, start angle = 0, end angle=180] .. controls (\i-1,0) and (\i,0) .. (\i,+1);
	\draw (\i+0.5,0) node {$X_\i$};
}

\foreach \i in {4,7}
	{
	\draw[thick,fill=red,fill opacity=0.65] (\i,+1) arc[x radius = 0.5, y radius = 0.5, start angle=180, end angle = 360] .. controls (\i+1,0) and (\i,0) .. (\i,-1) arc[x radius = 1.5, y radius = 0.3, start angle = 0, end angle = 180] .. controls (\i-2.5,0) and (\i-0.5,0) .. (\i,+1);
	\draw (\i-.35,0) node {$Z_\i$};
}
\draw[thick] (0,+1.2) --(0,+0.8);
\draw[thick] (0,-1.2) --(0,-0.8);
\draw[thick] (8,+1.2) --(8,+0.8);
\draw[thick] (8,-1.2) --(8,-0.8);
\draw (-.5,+1) node {$+$};
\draw (-.5,-1) node {$-$};
\end{tikzpicture}
\caption{Schematic example showing some of the terms in an LR code Hamiltonian with $G=\mbbZ/8\mbbZ$, $S_1=\{0,1\}$ and $S_2=\{0,3\}$.  The figure should be understood to have periodic boundary conditions.  Note that any time a green $X$ cell and a red $Z$ cell share a vertex, they in fact share an even number of vertices.}
\label{f:example}
\end{figure}

For a less-ordered example, we present a code in a finite cyclic group with more interesting degeneracy.  Let $n = 2^k-1$ for some integer $k > 0$, and our code will be in $\mbbZ/n\mbbZ$. Set
\[ S_1 = \overline{S_2} = \{0, 1, 3, 7, \ldots, 2^{k-1} - 1\} \]
Note that in this construction, the $gS_1$ block and $\overline{S_2}g$ block are identical. Then the rank of $M_G$ is $2 \cdot 2^{k-1}$, and so the degeneracy is $4^{2^{k-1}-1}$. We have a linear dependence between the following columns: $Z_0, Z_1, Z_2, \ldots, Z_{2^{k-1}}$. Then for $1 \leq i < j < k$, columns $i$ and $j$ each have a nonzero entry at $2^i+2^j -1$. Meanwhile, columns 0 and $i$ have nonzero entries at $2^i - 1$. Thus for column $Z_{2^i}$, for $i \neq j$, the entry at $2^i + 2^j - 1$ is canceled by column $Z_{2^j}$, and the entry at $2^{i+1} - 1$ is canceled by column $Z_0$.

Furthermore, for any $0 \leq l < 2^{k-1}-1$, $Z_{l + 2^{k-1}}$ can be written as a linear combination of $Z_{l}, Z_{l+1}, Z_{l+2}, \ldots, Z_{l+2^{k-1}}$. Thus, only vectors $Z_0$ through $Z_{2^{k-1}-1}$ are linearly independent, so
\[ \rank \left(\begin{array}{c} gS_1 \\[0.5em] \overline{S_2}g \end{array}\right) = 2^{k-1} \]
and $\rank M_G = 2 \cdot 2^{k-1}$.

For example, in $\mbbZ/31\mbbZ$, $S_1 = \overline{S_2} = \{0, 1, 3, 7, 15\}$
and the following chart shows the locations of nonzero entries of the vectors $Z_0, Z_1, \ldots, Z_{16}$:
\[ \begin{array}{cccccc}
     Z_0 &  Z_1 &  Z_2 &  Z_4 &  Z_8 & Z_{16}  \\ \hline
     0 &  1 &  2 &  4 &  8 & 16  \\
     1 &  2 &  3 &  5 &  9 & 17  \\
     3 &  4 &  5 &  7 & 11 & 19  \\
     7 &  8 &  9 & 11 & 15 & 23  \\
    15 & 16 & 17 & 19 & 23 &  0  \\
\end{array} \]
Excepting $Z_0$, the remaining $k$ vectors form a symmetric square matrix. The diagonal entries are canceled by the entries of $Z_0$.

\subsubsection{Dihedral groups}
Now consider the dihedral group of order $2n$
\[ G=D_{n} = \langle r, s \mid r^n = s^2=1, srs=r^{-1} \rangle \]
with $S_1=\{1,r\}$ and $S_2=\{1,s\}$.  Then the subgroup of rotations $\langle r \rangle$ is normal, so $|\langle s \rangle \backslash G / \langle r \rangle| = |\langle r,s \rangle \backslash G|$ and $\langle r,s\rangle = G$ so the degeneracy of $H_G$ is 4.

We have seen by some of the previous examples that normality of the subgroup generated by $S_1$ makes computing degeneracy easier.  The following result is another way in which this is true.

\begin{lem}
Consider a LR model where $G$ is any finite group and $S_1=S_2=N$, where $N$ is some normal subgroup of $G$. The ground state degeneracy of $H_G$ is
\[ 4^{|G| - |G/N|} \]
\end{lem}

\begin{proof}
Following the notation of the proof of Proposition \ref{prop:1s1t}, $Z_g$ has ones in precisely the spots $(gn,+)$ and $(ng,-)$ for $n \in N$. Then if $g_1,g_2 \in G$ with $g_1 g_2^{-1} \in N$, then $Z_{g_1} = Z_{g_2}$. Thus, for each coset of $G/N$, there is a single independent stabilizer $Z$ and so
\[ \rank \left(\begin{array}{c} gS_1 \\[0.5em] \overline{S_2}g \end{array}\right) = |G/N| \]
\end{proof}

As a consequence of this lemma, we can find codes with high degeneracy without requiring the parities of $|S_1|$ and $|S_2|$ to be even. For example, for odd $n$, the rotation subgroup $\langle r \rangle$ has odd order, and taking $S_1=S_2=\langle r \rangle$ yields a code with $4n-4$ logical qubits.

\subsubsection{Toric code}

Normal subgroups of $\mbbZ^2$ are of the form $m\mbbZ \times l\mbbZ$, which are exactly the usual periodic boundary conditions for doubly periodic functions on $\mbbR^2$.  The quotient groups $\mbbZ/{m\mbbZ}\times \mbbZ/{l\mbbZ}=\mbbZ^2/(m\mbbZ \times l\mbbZ)$ are the usual lattice tori.  It is a general theorem that any topological order realized by translation invariant Pauli Hamiltonian schemas of CSS form is equivalent to some copies of the toric code \cite{Ha16}.

Concretely, we can recover the toric code from LR models as follows.  The number of edges of the lattice torus $\mbbZ/{m\mbbZ}\times \mbbZ/{l\mbbZ}$
is exactly twice the number of vertices.  In the standard formulation of the toric code, the qubits live on the edges or bonds of the lattice.  We can describe this equivalently by using bi-qubits at the vertices: let a qubit on a ``vertical" edge become the $+$ qubit on the vertex at the bottom of the edge, and let a qubit on a ``horizontal" edge become the $-$ qubit on the vertex at the left of the edge.  Then, if $x,y$ denote generators of $\mbbZ^2$, choosing $S_1=\{1,x\}, S_2=\{1,-y\}$ in the LR model recovers exactly the toric code.

\subsubsection{Haah codes}

The Haah codes can be recovered by choosing
\[ S_1=\{1,x,y,z\} \text{ and } S_2=\{1,xy,yz,xz\} \]
where $x,y,z$ are the generators of $\mbbZ^3$, and then looking at the image of $S_1$ and $S_2$ inside a finite quotient of $\mbbZ^3$.  The group $\mbbZ^3$ is the fundamental group of the three torus $T^3$, whose universal cover is the Euclidean three space $\mathbb{R}^3$.  Finite sheeted covers of $T^3$ are always homeomorphic to $T^3$, which is one place to interpret where the Haah codes with periodic boundary conditions live.

It would be interesting to study other choices of $S_i$. For example, are there appropriate choices of $S_i$ that result in two copies of the 3D toric code?

\subsection{Excitations}

Haah codes are called type II fracton models.  In fracton models, point-like  excitations are in general sub-dimensional in the sense that they can move only on lower dimension sub-manifolds of the ambient space manifold. In Haah codes, minimal point-like excitations have the shape of tetrahedra, so it is natural to speculate that minimal point-like excitations in our generalizations would have the shape of the two subsets $S_1$ and $S_2$.  In the following, we will see that this speculation is wrong; how the shape of the minimal point-like excitations depends on $S_1$ and $S_2$ is more complicated.

\subsubsection{Spectrum for Abelian Groups with $|S_1|=|S_2|=2$}

\begin{prop}\label{prop:spectrum}
    If $G$ is a finite abelian group, with $S_1 = \{0,s\}$ and $S_2 = \{0,t\}$ such that $G = \langle s, t \rangle$, then the dimension of the eigenspace for eigenvalue $2i$, $0 \leq 2i \leq |G|$ is
    \[ 4 \cdot \sum_{a+b=i} \binom{n}{2a} \binom{n}{2b} \]
    The dimension of every odd eigenspace is 0.
\end{prop}

\begin{proof}
    For odd eigenvalues, the result follows from Proposition \ref{prop:hamiltonian}.

    For even eigenvalues, we observe that particles come in pairs, and there are two types corresponding to whether we are applying Pauli X or Pauli Z at a location. $\langle s, t \rangle = G$ implies that Pauli X's or Z's can transport any single particle to any other location. Thus, we can pick any even number of locations for each type of particle, and this will have total energy equal to the total number of particles, giving us the sum over binomial coefficients. The factor 4 is since our ground state degeneracy is 4 (Proposition \ref{prop:1s1t}), and these operators send the ground states to distinct excited states.
\end{proof}

\subsubsection{Minimal Excitation}

There is no simple resolution to the question: is there a relationship between the sizes of the sets $S_1$ and $S_2$, and energy of the minimal excitation? In Proposition \ref{prop:1s1t} and Proposition \ref{prop:spectrum}, we found that if $|S_1|=|S_2|=2$, then the ground state degeneracy is 4, and the minimal excitation has energy 2.

However, we can take a simple example in $G = D_3$ (the dihedral group on 6 elements), with $S_1 = \overline{S_2} = \{1, r, s\}$ where $s$ is a reflection and $r$ is a generator of the rotation subgroup of order $3$. A straightforward computation using $M_G$ says that this code has ground state degeneracy 4. However, the minimal excitation, despite having sets of size 3 and ground state degeneracy 4, has energy 1.

We can confirm this by hitting $(sr,+)$, $(sr,-)$, and $(sr^2,-)$ with Pauli X's. This set is incident to each $Z_{s}, Z_{sr}, Z_{r}, Z_{sr^2}$ twice, and incident to $Z_r$ once -- thus after these 3 operations, there is a single particle at $r$.

\section{LR models on general groups}

One of the salient features of the Haah codes is the seemingly random pattern of ground state degeneracies for various boundary conditions.  In this section, we organize the various boundary conditions as the profinite completion of the group, and, in the next section, we define a profinite low energy limit of the LR models.  A natural way to define a quantum system on any group $\Gamma$ using the quotients of all its FINs $N$ is to take a limit over them.

\subsection{Directed set of FINs and profinite completion}

A \emph{directed set} is a partially ordered set (poset) $(I, <)$ such that for every $\alpha, \beta \in I$, there exists a $\gamma\in I$ such that $\alpha, \beta \leq \gamma$.  If $I$ is a directed set, then an \emph{inverse system} over $I$ in a category is a family of objects $\{X_\alpha\}_{\alpha \in I}$ and a family of morphisms $f_{\beta \alpha }: X_\beta \rightarrow X_\alpha$ whenever $\alpha \leq \beta$, such that:
\begin{enumerate}
\item $f_{\alpha \alpha}=\textrm{id};$
\item $f_{\gamma \beta}f_{\beta \alpha}=f_{\gamma \alpha}$, whenever $\alpha \leq \beta \leq \gamma.$
\end{enumerate}
An inverse system as above will be denoted as $(X_\alpha, f_{\beta \alpha}, I)$.

\begin{defn}
The inverse (or projective) limit of an inverse system $(X_\alpha, f_{\beta \alpha}, I)$ is the set
\[ \ulalim X_\alpha=\left\{(x_\alpha)\in \prod_{\alpha \in I}X_\alpha \mid f_{\beta \alpha} (x_\beta)=x_\alpha \textrm{ whenever}\; \alpha \leq \beta \right\}.\]
\end{defn}

Given any group $\Gamma$, the collection of all of its FINs
\[ \mathcal{N}(\Gamma)=\{N \mid N\unlhd \Gamma, |\Gamma/N| < \infty\} \]
forms a directed set where the order relation $M \leq_d N$ is given by {\em reverse} inclusion $M \supseteq N$.  We note that $\mathcal{N}(\Gamma)$ is indeed a directed set since, given two FINs $M,N \in \mathcal{N}(\Gamma)$, they both contain their intersection $N\cap M$, which is a FIN of $\Gamma$.

Given $N\in \mathcal{N}(\Gamma)$, we let
\[G_N = \Gamma/N \]
denote the finite quotient of $\Gamma$ by $N$.  If $N\subseteq M$, then there is a short exact sequence
\[ 1\rightarrow M/N \rightarrow G_N \rightarrow G_M \rightarrow 1,\]
where the map $\rho_{N,M}: G_N \rightarrow G_M$ is the natural projection.  In particular, the set of finite quotients of $\Gamma$ forms an inverse system of finite groups over the directed set $\mathcal{N}(\Gamma)$.  We denote this inverse system by $\textrm{FIN}(\Gamma)=(G_N, \rho_{N,M}, \mathcal{N}(\Gamma))$.

(In words, we read $\textrm{FIN}(\Gamma)$ as the ``inverse system of finite quotients of $\Gamma$," so that we use FIN in two separate but dual ways: FINite quotient when talking about $G_N$ vs. Finite Index Normal when talking about $N$.)

\begin{defn}
The projective limit of $\textrm{FIN}(\Gamma)=(N, \rho_{N,M}, \mathcal{N}(\Gamma))$ is called the \emph{profinite completion} $\widehat{\Gamma}$ of $\Gamma$.
\end{defn}

Concretely, the profinite completion is
\[ \widehat{\Gamma}=\left\{(g_N N)_N \in \prod_{N\in \textrm{FIN}(\Gamma)} \Gamma/N \,\middle|\, \rho_{N,M}(g_N N)=(g_M M)\right\}. \]
That is, an element of $\widehat{\Gamma}$ is a sequence of elements in the finite quotients $\Gamma/N$ that satisfies compatibility with respect to $\rho_{N,M}$.  There is a natural homomorphism
\[ \begin{aligned}
j: \Gamma &\rightarrow \widehat{\Gamma} \\
\gamma &\mapsto (\gamma N)_N . \end{aligned} \]
Importantly, $j$ is not necessarily an injection.  The map $j$ is an embedding if and only if $\Gamma$ is residually finite.

An infinite sequence of nested FINs
\[ N_1 \supset N_2 \supset N_3 \supset \cdots \]
will be called a \emph{FIN sequence}.  If $\cap_{i\in \mathbb{N}}N_i=1$, we say the sequence is \emph{cofinal}.  If $\Gamma$ is not residually finite, then it does not have any cofinal FIN sequence.

The profinite completion $\widehat{\Gamma}$ has a natural topology---called the profinite topology---that makes it into a totally disconnected, compact topological group: each finite quotient $\Gamma/N$ is equipped with the discrete topology, the product of all of the $\Gamma/N$ is equipped with the product topology, and $\widehat{\Gamma}$ inherits a topology as a subspace.  We note that the product
\[ \prod_{N\in \textrm{FIN}(\Gamma)} \Gamma/N\]
is homeomorphic to one of the most famous fractals: the Cantor set.  Thus, being a subset of a fractal, $\widehat{\Gamma}$ is always inherently fractal; often $\widehat{\Gamma}$ itself has the topology of a Cantor set.

While $j$ need not be injective, the image $j(\Gamma)$ is always dense in $\widehat{\Gamma}$.  Therefore, the profinite completion $\widehat{\Gamma}$ can be thought of as an analogue of taking the closure of a set by including all limit points of sequences.

\subsubsection{Example: profinite completion of the integers $\mbbZ$}
When $\Gamma=\mbbZ$, the partial order on $\mathcal{N}(\mbbZ)$ is defined by division.  That is, $m\mbbZ \leq_d n\mbbZ$ if and only if $m|n$.  For example $p<_d p^2$ as $p^2 \mathbb{Z}\subset p \mathbb{Z} \subset \mathbb{Z}$ for a prime $p$.

By definition
\[ \begin{aligned}
\widehat{\mbbZ}&=\ulalim_{n}\mbbZ/n\mbbZ \\
&=\left\{(a_n)^{\infty}_{n=1}\in \prod_{n=1}^\infty (\mbbZ/n\mbbZ) \mid a_m \equiv a_n \pmod*{n}, \forall n|m\right\}.
\end{aligned} \]
An integer $a$ is included in $\widehat{\mbbZ}$ as $j(a) = (a_n)$, where $a_n$ is its reduction modulo $n$.  In fact, by basically the Chinese remainder theorem,
\[\widehat{\mbbZ}\cong \prod_{p\;\textrm{prime}}\mbbZ_p,\]\footnote{In this paper we follow the number theoretical convention that the finite cyclic group with $p$ elements of $\mathbb{Z}$ mod $p$ is denoted as $\mbbZ/p\mbbZ$, while the $p$-adic integers as $\mathbb{Z}_p$.}the product of $p$-adic integers for all distinct prime $p$.  The fractal nature of $\widehat{\mbbZ}$ is manifest by the fact that in the profinite topology, the $p$-adic integers $\mathbb{Z}_p$ are a Cantor set.

\subsection{Inverse system of coverings}

Given a connected manifold $Y$, its universal covering $\widetilde{Y}$ can be constructed as follows.  Fix a point $y_0\in Y$ and for any point $y\in Y$, let $F_y$ be the set of paths from $y_0$ to $y$ up to relative homotopy.  Note that as a set, $F_y$ can be identified with $\pi_1(Y,y_0)$ by fixing a path in $F_y$.  Then with an appropriate topology, the union $\{F_y\}_{y\in Y}$ is the universal cover $\widetilde{Y}$ of $Y$.

Given a FIN $N \subset \pi_1(Y)$, the corresponding covering space $\widetilde{Y}_N$ with fundamental group $N$ can be constructed by introducing an equivalence relation into the set $F_y$. Two classes $\gamma_1$ and $\gamma_2$ in $F_y$ are equivalent if $\gamma_1\cdot \bar{\gamma_2}\in N$.  The resulting quotient space is then $\widetilde{Y}_N$.

The collection $(\widetilde{Y}_N, f_{N,M}, \mathcal{N}(\Gamma))$ forms an inverse system $\mathcal{C}(Y)$ of regular coverings corresponding to $\textrm{FIN}(\Gamma)$, where the map $f_{N,M}$ induces the map $\rho_{N,M}$ in $\textrm{FIN}(\pi_1(Y))$.

\subsubsection{Galois towers}

The topological counterparts of FIN sequences in $\textrm{Fin}(\pi_1(Y))$ are Galois towers in $\mathcal{C}(Y)$.

Let $X$ be any topological space and $$\cdots \rightarrow \tilde{X}_n \rightarrow \tilde{X}_{n-1}\rightarrow \cdots \rightarrow \tilde{X_1}\rightarrow X=X_0$$ be a sequence of pairwise regular covers of $X$.  The sequence of regular covers $\{X_n\}$ of $X$ will be called a Galois tower of $X$, denoted as $GT(X)=\{X_n\}$. A Galois tower is cofinal if $\cap_n \pi_1(X_n)=1$.

Set $G_n=\pi_1(X_n)/\pi_1(X_{n-1})$, then $$G_0\leftarrow G_1\leftarrow G_2\leftarrow ...$$ is a sequence of finite groups with $G_0=\pi_1(X)$.  The sequence of groups $\{G_n\}$ will be called the group of the $GT(X)$, denoted as $\Pi(X)=\{G_n\}$.

Given a topological space $X$ and two subsets $S_1, S_2$ of its fundamental group $\pi_1(X)$.
Then $S_1,S_2$ descend to two subsets $S_1^n, S_2^n$ of each $G_n$.  In the following, we will often drop the superscript $n$ from $S_1^n, S_2^n$.

\subsection{Direct system of vector spaces}

A direct system in a category with a direct set $I$ is a family of objects $\{X_\alpha\}_{\alpha \in I}$ and a family of morphisms $f_{\alpha \beta}: X_\alpha \rightarrow X_\beta$ whenever $\alpha \leq \beta$, such that:
\begin{enumerate}
\item $f_{\alpha \alpha}=\textrm{id};$
\item $f_{\alpha \beta}f_{\beta \gamma}=f_{\alpha \gamma}$, whenever $\alpha \leq \beta \leq \gamma.$
\end{enumerate}

\begin{defn}
Given a direct system of vector spaces $(V_\alpha, f_{\alpha \beta}, I)$.  The direct (or inductive) limit is defined as the quotient vector space
$$\lim_{\to} V_\alpha={\widetilde{V}}/K,$$
where $\widetilde{V}$ and $K$ are the vector spaces $\widetilde{V}=\oplus_{\alpha \in I}V_\alpha$, and $K$ spanned by all vectors $\{x_\beta-f_{\alpha \beta}(x_\alpha)\}$ whenever $\alpha < \beta$.
\end{defn}

If the vector spaces $V_\alpha$ are Hilbert spaces and $f_{\alpha \beta}$ are unitaries, then $\lim_{\to} V_\alpha$ becomes a Hilbert space.  Given any two vectors $[v],[w]\in \lim_{\to} V_\alpha$, there exist $\alpha, \beta$ with representatives $v\in V_\alpha, w\in V_\beta$.  Choose a $\gamma \in I$ such that $\alpha, \beta \leq \gamma$, then define the inner product of $[v],[w]$ in $V_\gamma$.

We have the freedom to make $\mathcal{N}(\Gamma)$ either into a direct system or an inverse system.  TQFTs are representations of bordism categories, hence should exchange limit and colimit.  The target Hilbert space from our construction will be a direct limit.  Since the LR models are defined on the quotients $\Gamma/N$, the natural choice is the inverse system because $N\subset M$ should lead to some map $L(\Gamma/M,d,q)\rightarrow L(\Gamma/N,d,q)$.

\subsection{Virtually invariant total Hilbert space}

Given a pair of FINs $N\subset M$, there is an isometric embedding  $\iota: L(\Gamma/M,d,q)\rightarrow L(\Gamma/N, d,q)$, where the image of $L(\Gamma/N,d,q)$ is a subspace of the $M/N$ invariant states.  The direct limit $L_{\iota}(\Gamma, d,q)$ of $(L(\Gamma/N, d,q),\iota_{N,M}, \mathcal{N}(\Gamma))$ is a possible definition for the total Hilbert space $L(\Gamma,d,q)$.

\subsection{GNS construction of the total Hilbert space}

In mathematical physics, the GNS construction in $\mbbC^{*}$-algebras is usually used to construct a Hilbert space from algebras, which would lead to another version of $L(\Gamma,d,q)$.

\section{Profinite low energy limits on 3-manifolds}

Three-manifolds are special for LR models as their fundamental groups are always residually finite, and they are powerful invariants as illustrated by the Poincare conjecture.  As a comparison, there are infinitely many simply connected closed $4$-manifolds such as $S^4, S^2\times S^2, {\mbbC}P^2$ and their connected sums.  Manifolds in this section are always connected. Otherwise the fundamental group should be replaced by the fundamental groupoid.

It is both interesting and challenging to define the LR model on general groups by taking appropriate limits so that LR models represent new phases of matter, presumably closely related to topological phases. The difficulty lies in identifying a proper way to perform some version of scaling or renormalization.  In conventional topological phases, low energy scaling limit is essentially trivial as the ground state manifold is independent of the lattices.  Since there are no Riemannian metrics involved, the large volume limit can be either regarded as trivial as the scaling limit or simply irrelevant.  Thinking a bit harder, we propose a version of infinite volume limit as the extension of a conventional TQFT to open manifolds---manifolds that are non-compact without boundaries such as $\mbbR^n$.  Conventional TQFTs are defined only for compact manifolds, which can be extended to open manifolds \cite{F97}.  There are much more interesting open manifolds than $\mbbR^n$ and a famous one is the Whitehead manifold $W$: a contractible $3$-manifold that is not homeomorphic to $\mbbR^3$, but $W\times \mbbR$ is homeomorphic to $\mbbR^4$.  The difference between $\mbbR^3$ and $W$ lies at the end---the neighborhoods at infinity.

Open manifolds can be regarded as limits of closed manifolds, though not in any canonical way.  One simple example would be $\mbbR^n$ as the limit of a sequence of spheres $S^n$ using the stereographic projections: a sequence of $S^n$ sitting at the origin in the upper half space with increasingly larger sizes that go to infinity.  Slightly non-trivial would be $\mbbR^n$ as the limit of a sequence of $n$-tori $T^n$ using the product of a sequence of circles limiting to $\mbbR$.  If such sequences are regarded as some kinds of scaling, then the embedding of one local Hilbert space into another depends heavily on the choice of the limiting sequence.  We propose to take the limit of the LR model on $3$-manifolds using profinite completions, which amounts to scaling via covering space sequences limitng to its universal cover.  It is not clear if such a limit should be called a scaling limit or a large volume limit, so we will simply refer it as the {\it profinite low energy limit}.

Since our model is defined on the fundamental group of a manifold, abstractly all manifolds with the same $\pi_1$ have the same theory.  But this is not completely accurate when locality is taken into consideration.  For example, both the circle $S^1$ and the 3-manifold $S^1\times S^2$ have the same fundamental group $\mbbZ$.  It is possible to imagine that the LDOF on $S^2$ can change how we arrive at a limit.

\subsection{3-manifolds with finite fundamental groups}

There are infinitely many 3-manifolds with finite fundamental groups such as the lens spaces $L(p,q)$, which all have the same universal cover $S^3$.

Since the profinite completion of a finite group $G$ is simply itself, the LR model on $G$ is already the full story.

The Hilbert space associated to $S^3$ is trivial ($\cong \mbbC$), which is important as it suggests the stability of the theory.  In conventional TQFTs,  the triviality of the Hilbert space for $S^n$ is conjectured to be equivalent to the stability of the TQFT.

\subsection{Limits}

Two possible definitions of the total Hilbert space $L(\Gamma,d,q)$ are suggested at the end of last section.  A potential problem is the mixing of different energy scales.  Since we are only interested in low energy effective theories, we will define the low energy limit following the framework \cite{SW18}.

The low energy limit in \cite{SW18} is defined through a double limit process.
Given a sequence of theories $\{W_\alpha\}=\{(L_\alpha, H_\alpha)\}$ consisting of pairs of Hilbert spaces $L_\alpha$ and Hamiltonians $H_\alpha$.  Each Hilbert space $L_\alpha$ decomposes into energy eigenspaces $L^\alpha_{\lambda_i}$ so that $L_\alpha=\oplus_{i}L^\alpha_{\lambda_i}$, where the energy levels are $\lambda_0<\lambda_1<\cdots .$  Assume the energy level sequence $\lambda^\alpha_i$ for each fixed $i$ converges to $\lambda^\infty_i$ as $\alpha$ goes to $\infty$ and the corresponding Hilbert spaces $L^\alpha_{\lambda_i}$ converging to $L^\infty_i$, then the low energy limit Hilbert space is $L^\infty=\oplus L^\infty_i$ if defined.

\subsubsection{Direct system of energy eigenspaces}

Given a group $\Gamma$ with two subsets $S_i$, and a FIN $N$.   The LR Hamiltonian $H(G_N;S_1,S_2)$ leads to a direct sum decomposition $$L(G_N,2,2)=\oplus_{i=0}L_i(G_N;S_1,S_2),$$ where $0=E_0<E^N_1<\cdots $ are the energy levels or eigenvalues of $H(G_N;S_1,S_2)$.  The energy level sequence $E^N_i$ consists of positive integers if $i>0$.

Suppose there exists an $E_i$-energy direct system of Hilbert spaces $(L_i(G_N;S_1,S_2), f_{M,N}, \mathcal{N}(\Gamma))$\footnote{There is an issue about the right choice of maps $f_{M,N}$ in general, which will be left to the future.} for each $i$, then the direct limit of $(L_i(G_N;S_1,S_2), f_{M,N}, \mathcal{N}(\Gamma))$ will be denoted as $L_i(\Gamma, S_1,S_2)$.

For the ground state direct system $(L_0(G_N;S_1,S_2), f_{M,N}, \mathcal{N}(\Gamma))$, one possible choice for the connecting maps $f_{N,M}$ is as follows.  Given a pair of FINs $N\subset M$,
$$ L_0(G_M;S_1,S_2)\stackrel{i_{M}}{\longrightarrow} L(G_M;S_1,S_2)\stackrel{\iota_{M,N}}{\longrightarrow} L(G_N;S_1,S_2)\stackrel{\pi_{N}}{\longrightarrow}L_0(G_N;S_1,S_2),$$
where $i_{M}$ is the inclusion, $\pi_{N}$ the projection, and $\iota_{M,N}$ the $M/N$-equivariantization map.

\subsubsection{Ground state manifold}

Given $(\Gamma, S_1,S_2)$, we define the ground state manifold in the profinite low energy limit as the direct limt of $(L_0(G_N;S_1,S_2), f_{M,N}, \mathcal{N}(\Gamma))$ if exists:
$$V(\Gamma,S_1,S_2)=\lim_{\substack{\to\\ N}} L_0(G_N,S_1,S_2).$$

The simplest example is $(\mbbZ, S_1=\{0,p\}, S_2=\{0,q\})$, where $p,q$ are distinct primes.  As shown earlier, the ground state manifold for $(G_N,S_1,S_2)$ is always $\mbbC^4$ independent of $n$ in $\mbbZ/n\mbbZ$.  It follows that the ground state manifold in the profinite low energy limit is $\mbbC^4$, too.

It would be very interesting to describe the ground state manifold in the profinite low energy limit of the Haah codes.

\subsubsection{Low energy limit}

Finally, the total Hilbert space of the profinite low energy limit for LR models is defined as $$L(\Gamma,2,2|S_1,S_2)=\oplus_{i=0}L_i(\Gamma, S_1,S_2).$$
The infinitely many Hilbert spaces cannot be added without a rescaling of the norms.  An obvious choice here is $p_i=\frac{e^{-E_i/kT}}{E}$, where $\sum_i e^{-E_i/kT}$ converges to $E$.

\subsection{Inside and outside of manifolds}

Our generalization of Haah codes to $3$-manifolds can be regarded as either living outside the manifolds in the universal covers or inside the manifolds. Outside the model lives on lattices of points, while inside on \lq\lq lattices" of closed loops.

\subsubsection{Topological point lattice model}

Given a closed manifold $Y$ with a fixed point $y_0\in Y$.  The preimage $\Pi(Y)$ of $y_0$ in the universal cover $\widetilde{Y}$ is in one-one correspondence with the fundamental group $\pi_1(Y,y_0)$, hence $\pi_1(Y)$ can be identified with $\Pi(Y)$.  Then the LR model on $\pi_1(Y)$ can be regarded as a model on $\Pi(Y)$.

The two subsets $S_1,S_2$ in the model can be identified as a collection of points around $y_0$, which serve as a unit cell of the lattice $\Pi(Y)$.

If $Y$ has a Riemannian metric, then the sizes of $gS_1\cup \bar{S_2}g$ and $S_2g\cup g\bar{S_1}$ could be unbounded.  The unbounded cases seem to be non-local in a sense, so it would be interesting to understand when the sizes of $gS_1\cup \bar{S_2}g$ and $S_2g\cup g\bar{S_1}$ are bounded.

\subsubsection{Topological loop lattice model}

Intrinsically, our model is defined on loops in $Y$. Let $\Delta(Y)$ be a cellulation of $Y$ and $T$ a maximal spanning tree of the $1$-skeleton of $\Delta(Y)$.  Then the homotopy group $\pi_1(Y,T)$ is isomorphic to $\pi_1(Y)$.  It follows that $\pi_1(Y)$ can be identified as closed loops based at $T$ up to homotopy.  We could either fix a representative set or add a term to the Hamiltonian $H{(G_n;S_1,S_2)}$ so that homotopic loops are in superpositions. Then each element of $\pi_1(Y)$ is represented either by a single closed loop or by a superposition of loops.

Given a FIN $N$, choose the cellulation of $Y$ with $|G_N|$ many vertices. Choosing representatives of $G_N$ in $\pi_1(Y)$ and identifying them with the vertex set, we draw a closed loop at each vertex $v$ to representing the fundamental group element $g_v$.  If $Y$ is dim=$3$ or higher, the loops can be chosen to be embedded to form a link $L_N$ in $Y$. Therefore, our models live on such links $L_N$ inside $Y$.

The LDOF on each loop is a bi-qubit.  We can realize the bi-qubit by putting a toric code on each loop thickening slightly to a torus.  Hence one interpretation of the Haah codes is that they emerge from links of toric codes in 3-manifolds.  It follows that the bi-qubit LDOF is already not spatially local as it spreads over a loop or torus, so locality of the theory is a subtle question.

\subsection{More examples}

The example $(\mbbZ, (0,p), (0,q))$ can be thought as a LR model on $S^1\times S^2$.  The fundamental group $RP^3\# RP^3$ is the free product of $\mbbZ_2$ with itself---the infinite dihedral group.  One Galois tower would be coverings with dihedral groups $D_n=\{r,s|r^n,s^2\}$ as deck transformations.  More interesting would be Galois towers for hyperbolic $3$-manifolds and Euclidean manifolds called platycosms.

\section{Future directions}

\subsection{Profinite invariant of groups}

One application is to use the Hilbert space $V(\Gamma)$ as a profinite invariant of the group $\Gamma$.  Then group theoretical properties of the group $\Gamma$ such as residual finiteness or LERFness (locally extended residual finiteness) should have manifestations in the Hilbert space $V(\Gamma)$.

A new knot invariant can be defined using our model on knot fundamental groups.  Taking the knot group of a knot $K$ with the two subsets $S_1=\{1,m\},S_2=\{1,l\}$, where $m,l$ are the meridian and longitude of $K$, we obtain the ground state Hilbert space $V(\pi_1(S^3\backslash K), S_1,S_2)$ as a knot invariant.

\subsection{Quantum codes}

The motivation for \cite{Hs14} was to find quantum codes with better asymptotic properties.  It would be interesting to study the distance of the quantum codes from some Galois towers.

To find interesting codes, it seems that we need to choose the two subsets $S_1$ and $S_2$ to be somehow independent.  In particular, it would be interesting to understand conditions on $S_1$ and $S_2$ that lead to codes with fast growth distance.

\subsection{Loop statistics}

While there are no interesting particle statistics in Haah codes, there are still loop statistics \cite{Ha11}.  It would be interesting to study them explicitly, which could provide hint for a framework to understand type II fractons.

\subsection{Intrinsic fracton models}

The fracton models that are studied so far are simple Ising-like models after ungauging.  More interesting fracton models would be those that possess non-trivial topological order after being ungauged.

\subsection{Tension between virtual invariance and low energy}

A property of a space that holds up to taking finite-index covering spaces is usually referred to as \emph{virtual} in topology.

Given $N\subset M$, the natural embedding $\iota: L(\Gamma/M,d,q)\rightarrow L(\Gamma/N,d,q)$ is not compatible with the Hamiltonians, i.e. $\iota$ does not commute with the two Hamiltonians.  Therefore, the equivariant embeddings mixed states of different energies.  But it could be used to construct a total Hilbert space.

Given a cofinal FIN sequence $N_i$ of a group $\Gamma$, then a ground state in each $L(G_{N_i},S_1,S_2)$ defines a tracial state in the $\mbbC^{*}$-algebra of $L(G_{N_i},S_1,S_2)$.  The GNS construction leads to a Hilbert space, so potentially a version of $L(\Gamma, d,q)=\otimes_{\gamma \in \Gamma} ((\mbbC^d)^{\otimes q})_\gamma.$  We can also define some limiting Hamiltonian.  But it is not clear how this quantum system is related to our profinite low energy limit.

\end{document}